\documentclass{amsart}
\pdfoutput=1
\usepackage{amssymb,latexsym}
\usepackage{amstext,amsthm,amscd}
\usepackage{amsmath}
\usepackage{amsgen,amsfonts,amsbsy}
\usepackage{tikz-cd}
\usepackage{mathrsfs}

\theoremstyle{plain}
\newtheorem{theorem}{Theorem}[section]
\newtheorem{proposition}[theorem]{Proposition}
\newtheorem{lemma}[theorem]{Lemma}

\theoremstyle{definition}
\newtheorem{definition}[theorem]{Definition}

\newtheorem{remark}[theorem]{Remark}

\newcommand{\Z}{\mathbb{Z}}
\newcommand{\N}{\mathbb{N}}

\newcommand{\C}{\mathbb{C}}
\newcommand{\set}[2]{\{#1|\ #2\}}
\newcommand{\sub}{\subseteq}

\newcommand{\lcm}{\mathrm{lcm}}

\newcommand{\diag}{\mathrm{diag}}
\newcommand{\U}{\mathrm{U}}
\newcommand{\T}{\Omega}
\newcommand{\CC}{\mathcal{C}}

\newcommand{\SL}{\mathrm{SL}}
\newcommand{\Sp}{\mathrm{Sp}}

\renewcommand{\H}{\mathcal{H}}

\newenvironment{mat2}[4]{\left(\begin{array}{cc} #1&#2\\#3&#4 \end{array}\right)}{}
\newenvironment{matt2}[4]{\left[\begin{array}{cc} #1&#2\\#3&#4 \end{array}\right]}{}
\newenvironment{smat2}[4]{\begin{smallmatrix}#1&#2\\#3&#4\end{smallmatrix}}{}

\begin{document}
\title[Structure of Clifford groups of composite finite quantum systems]{Structure of Clifford groups of composite finite quantum systems}
	 
\author[M.~Korbel\'a\v{r}]{Miroslav~Korbel\'a\v{r}}
\address{Department of Mathematics, Faculty of Electrical Engineering, 
	Czech Technical University in Prague, Jugosl\'avsk\'ych partiz\'an\r{u} 1, 
	166 27 Prague 6, Czech Republic}
\email{korbemir@fel.cvut.cz}

\author[J.~Tolar]{Ji\v{r}\'i Tolar}
\address{Department of Physics,
	Faculty of Nuclear Sciences and  Physical  Engineering,
	Czech Technical University in Prague, B\v{r}ehov\'{a} 7,
	115 19 Prague 1, Czech Republic}
\email{jiri.tolar@fjfi.cvut.cz}


 \begin{abstract}
In this paper the Clifford groups of general multipartite quantum systems with configuration space $\Z_{n_1}\oplus\cdots\oplus\Z_{n_k}$, $k\geq 1$ are studied. It is known that the Clifford group is a natural semidirect product
provided the dimension $N=n_1\cdots n_k$ of the corresponding Hilbert space  is an odd number.

For even $N$  special results on the Clifford groups are scattered
in the mathematical literature, but they do not concern the semidirect product structure.
Using the relation of generators of the associated symplectic group $\Sp_{[n_1,\dots,n_k]}$  we prove that for even $N=n_1\cdots n_k$ both Clifford group and the projective Clifford group are natural semidirect products if and only if $N$ is not divisible by four.

\end{abstract}

\maketitle \vspace{2ex}


\thanks{}



\section{Introduction}\label{introduction}

This paper generalizes the results obtained in our paper \cite{KorbTolar23} 
for the class of simple finite quantum systems -- qudits with cyclic groups $\Z_N$
as configuration spaces. See also an unpublished paper \cite{Hashimoto} 
where similar results were obtained by another method.

Our aim here is to consider all finite quantum systems with arbitrary finite abelian groups 
$\Z_{n_1}\oplus\cdots\oplus\Z_{n_k}$ as configuration spaces.
Mathematically, the Hilbert spaces of the corresponding quantum systems are constructed as 
tensor products of the qudit Hilbert spaces
 $$\mathscr{H} = \mathscr{H}_{n_1}\otimes\dots\otimes\mathscr{H}_{n_k},$$
i.e. the quantum system is a composite system and its subsystems are just simple qudits.
Let us note that, physically, any such tensor product structure makes sense, provided
that there is an experimental way to manipulate and probe the constituent quantum subsystems
by available interactions and measurements \cite{Zanardi}.

Our paper \cite{KorbTolar23} brought striking results concerning the structure of 
projective Clifford groups in even-dimensional Hilbert spaces. 
It provided an important result concerning a natural question (coming already from classical mechanics) --
is the Clifford group a semidirect product?
It is well known that the answer is positive if the dimension of Hilbert space is odd
\cite{Gross}. However, the answer was not apparent in even dimensions \cite{Appleby}. 

In \cite{KorbTolar23} the simple case of projective Clifford group 
corresponding to one qudit of even dimension $N$ was considered.
We proved that for dimensions divisible by four,
projective Clifford groups are \emph{not} semidirect products, 
but for dimensions $N$ with $N/2$ odd they, surprisingly, are.
Our proof was based on appropriate group presentations of the groups 
$\SL(2,\Z_N)$ and $\SL(2,\Z_{N})\ltimes \Z_{N}^2$, and
a suitable description of the projective Clifford group.

However, for composite systems of even dimensions such an approach is not suitable 
for direct extension of our results to projective Clifford groups.
But let us note that in the last section of our paper \cite{KorbTolar23} 
a conjecture was suggested concerning the general case of finite-dimensional quantum mechanics
with Hilbert space of even dimension,
i.e., when the configuration space is an arbitrary finite abelian group
$A=\Z_{n_1}\oplus\cdots\oplus\Z_{n_k}$ of even order.

\

\textbf{Conjecture:} Let the configuration space of finite-dimensional quantum mechanics
be a finite abelian group  $A$ of an even order,
i.e. the cardinality $|A|$ of $A$ is an even number.
Then the following are equivalent:
\begin{itemize}
	\item The corresponding Clifford group is a  semidirect product 
	       related to its natural exact sequence with its Heisenberg group.
	\item The corresponding \emph{projective} Clifford group is a  semidirect product 
	       related to its natural exact sequence with its \emph{projective} Heisenberg group.
	\item $|A|=2\ (\mathrm{mod}\ 4)$.
\end{itemize}

In the present paper we prove this Conjecture, as our main result (Theorems \ref{main_theorem_1} and \ref{main_theorem_2}).

\

It is a well-known fact that every finite abelian group is isomorphic to a direct product (or direct sum) of finite cyclic groups.
Instead of working with a finite abstract abelian group, we will rather consider a fixed cartesian product
of a finite set of finite cyclic groups.

In fact, in \cite{StovTolar84} a complete description of quantum kinematics
was presented for the class of systems
whose underlying configuration spaces are finite sets equipped with
the structure of finite Abelian groups.
\footnote{
The paper \cite{StovTolar84} was the first where Schwinger's incomplete proposal
of tensor product decomposition was corrected.}

\

In this paper, the presented Conjecture will be proved by combining  an idea proposed
in \cite{Hashimoto} and the approach used in \cite{KorbTolar23} where a simple qudit of dimension $N$ was dealt with. Our current method is based on the detailed description of the structure and generators of symplectic groups associated to the multipartite systems that was already investigated in our previous work \cite{KorbTolar10,KorbTolar12}.\footnote{
In the course of finishing the manuscript we found a paper by C. Galindo
\cite{Galindo} in which our Conjecture is proved by using methods of
finite group cohomology.}

\section{Clifford and symplectic groups for the multipartite systems}

Let us briefly recall the definition of Clifford group and its associated symplectic group.

Let $P_n$ and $Q_n$ be the generalized unitary Pauli matrices of the operators of momentum and position. They fulfill the equality $Q_nP_n=\omega_n P_nQ_n$ where $\omega_n=e^{\frac{2\pi\mathbf{i}}{n}}\in\C$. Let $\U(n)$ be the group of all unitary $n\times n$ complex matrices and $I_n$ denote the identity $n\times n$ matrix.

\

For a multipartite system with the the abelian group $\Z_{n_1}\oplus\cdots\oplus\Z_{n_k}$  as the configuration space, where $k\geq 1$ and $n_1,\dots,n_k\geq 2$ are positive integers, we define  the \emph{Heisenberg group}
$$\H_{(n_1,\dots,n_k)}=\set{\lambda\cdot P_{n_1}^{i_1}Q_{n_1}^{j_1}\otimes\cdots\otimes
P_{n_k}^{i_k}Q_{n_k}^{j_k}}{i_1,j_1,\dots,i_k,j_k\in\Z, \lambda\in\U(1)}$$
  of Kronecker product of generalized Pauli matrices and the \emph{Clifford group}  $$\CC_{(n_1,\dots,n_k)}=N_{\U(N)}(\H_{(n_1,\dots,n_k)})$$ as the normalizer of the Heiseberg group $\H_{(n_1,\dots,n_k)}$  in the unitary group $\U(N)$,  where $N=n_1\cdots n_k$.

An element $U\in\CC_{(n_1,\dots,n_k)}$ has a natural action on $\H_{(n_1,\dots,n_k)}$ as an automorphism  of the form  $X\mapsto UXU^{-1}$ for $X\in \H_{(n_1,\dots,n_k)}$. More generally, for a global linear group $GL(n,\C)$ we denote by $Ad_M$ its \emph{inner automorphism} associated to $M\in GL(n,\C)$, i.e. $Ad_M(X)=MXM^{-1}$ for every $X\in GL(n,\C)$.

Since $\H_{(n_1,\dots,n_k)}$ contains an (ortonormal) basis of the vector space of all complex $N\times N$ matrices (see e.g. \cite{AschChilWoc}), it follows that for $U,U'\in GL(N,\C)$ the maps $Ad_U$ and $Ad_{U'}$ coincide on $\H_{(n_1,\dots,n_k)}$  if and only if $U=\lambda U'$ for some $\lambda\in\C\setminus\{0\}$ (we denote this as $U\sim U'$), in other words if $U$ and $U'$ are equal as projective linear transformations (of the projective space $P(\C^N)$).

Since these Clifford operations, i.e. the corresponding inner automorphisms,
are not affected by phase factors, it is natural to study the
\emph{projective Heisenberg group}
$$\overline{\H}_{(n_1,\dots,n_k)}=\set{Ad_{M}}{M\in\H_{(n_1,\dots,n_k)}}$$
and the \emph{projective Clifford group}
$$\overline{\CC}_{(n_1,\dots,n_k)}=\set{Ad_{M}}{M\in\CC_{(n_1,\dots,n_k)}}.$$

We can easily see that $\overline{\CC}_{(n_1,\dots,n_k)}$  is the normalizer of $\overline{\H}_{(n_1,\dots,n_k)}$ in the group $\overline{U}(N)=\set{Ad_{M}}{M\in\U(N)}$, i.e. $\overline{\CC}_{(n_1,\dots,n_k)}=N_{\overline{U}(N)}(\overline{\H}_{(n_1,\dots,n_k)})$, and that there is an isomorphism of the factor-groups $\CC_{(n_1,\dots,n_k)}/\H_{(n_1,\dots,n_k)}\cong\overline{\CC}_{(n_1,\dots,n_k)}/\overline{\H}_{(n_1,\dots,n_k)}$ of the form $M\cdot \H_{(n_1,\dots,n_k)}\mapsto Ad_{M}\cdot \overline{\H}_{(n_1,\dots,n_k)}$.

\begin{remark}
Let us consider the projective global linear group $\overline{GL}(N,\C))=\set{Ad_M}{M\in GL(N,\C)}$. In \cite{KorbTolar12} a more general normalizer than $N_{\overline{U}(N)}(\overline{\H}_{(n_1,\dots,n_k)})$
was studied, namely the normalizer $N_{\overline{GL}(N,\C)}(\overline{\H}_{(n_1,\dots,n_k)})$ (see \cite[Definition 5.1]{KorbTolar12}). However, it turned out that $$N_{\overline{GL}(N,\C)}(\overline{\H}_{(n_1,\dots,n_k)})\sub \overline{U}(N)$$
(this follows from the generating set of this normalizer, see the results \cite[Theorem 5.9]{KorbTolar12} and \cite[Corollary 11]{HPPT02}) and that means that both these two normalizers are, in fact, equal. Hence $$\overline{\CC}_{(n_1,\dots,n_k)}=N_{\overline{GL}(N,\C)}(\overline{\H}_{(n_1,\dots,n_k)})$$ and we may therefore apply the results from \cite{KorbTolar12} for the projective Clifford group defined above.
\end{remark}

Let us now recall the main theorem from \cite{KorbTolar12} that states that the factor-group  $\CC_{(n_1,\dots,n_k)}/\H_{(n_1,\dots,n_k)}\cong\overline{\CC}_{(n_1,\dots,n_k)}/\overline{\H}_{(n_1,\dots,n_k)}$ is isomorphic to a certain symplectic group $\Sp_{[n_1,\dots,n_k]}$ (see Definition $\ref{def_11}$ below).

\begin{theorem}\cite[Theorem 5.9]{KorbTolar12}\label{theorem_1}

Let $k\geq 1$ and $n_1,\dots,n_k\geq 2$ be positive integes. Then:

For every $Ad_{U}\in \overline{\CC}_{(n_1,\dots,n_k)}$, with $U\in \CC_{(n_1,\dots,n_k)}$, there is a unique (matrix equivalence class) $[\mathbb{H}]\in \Sp_{[n_1,\dots,n_k]}$ such that for all $i_1,j_1,\dots,i_k,j_k\in\Z$ there holds
$$Ad_{U}\left(P_{n_1}^{i_1}Q_{n_1}^{j_1}\otimes\cdots\otimes P_{n_k}^{i_k}Q_{n_k}^{j_k}\right)\sim P_{n_1}^{i'_1}Q_{n_1}^{j'_1}\otimes\cdots\otimes P_{n_k}^{i'_k}Q_{n_k}^{j'_k}$$
where $\begin{pmatrix}i'_1\\ j'_1\\\vdots\\i'_k\\j'_k \end{pmatrix}=\mathbb{H}\begin{pmatrix}i_1\\ j_1\\\vdots\\i_k\\j_k \end{pmatrix}$.

 The map $\overline{\pi}: \overline{\CC}_{(n_1,\dots,n_k)}\to \Sp_{[n_1,\dots,n_k]}$, $Ad_{U}\mapsto[\mathbb{H}]$ is a group epimorphism with $\ker(\overline{\pi})=\overline{\H}_{(n_1,\dots,n_k)}$.

In particular, there is a group isomorphism $$\overline{\CC}_{(n_1,\dots,n_k)}/\overline{\H}_{(n_1,\dots,n_k)}\cong \Sp_{[n_1,\dots,n_k]}.$$
\end{theorem}

\

\begin{definition}\label{def_11}(\textbf{of the symplectic group $\Sp_{[n_{1},\dots,n_{k}]}$})

Let $\mathcal{S}_{[n_{1},\dots,n_{k}]}$ be a monoid consisting of
$k\times k$ matrices $\mathbb{H}$ composed of $2\times 2$ blocks, where the block on position $(i,j)$, $i,j\in\{1,\dots,k\}$, has a form
  $$(\mathbb{H})_{ij}=\tfrac{n_{i}}{\gcd(n_{i},n_{j})}A_{ij}$$
with $A_{ij}$ beeing $2\times
2$ matrices over $\Z$. We consider $\mathcal{S}_{[n_{1},\dots,n_{k}]}$ as a monoid with the usual multiplication of matrices.

An adjoint $\mathbb{H}^{\ast}\in\mathcal{S}_{[n_{1},\dots,n_{k}]}$ of
$\mathbb{H}\in\mathcal{S}_{[n_{1},\dots,n_{k}]}$
 is then defined by
$$(\mathbb{H}^{\ast})_{ij}=\tfrac{n_{i}}{\gcd(n_{i},n_{j})}A_{ji}^{T}$$
for $i,j\in\{1,\dots,k\}$.

Using a diagonal matrix
 $$\mathbb{D}=\diag\Big(\tfrac{\lcm(n_{1},\dots,n_{k})}{n_{1}}I_{2},
 \dots,\tfrac{\lcm(n_{1},\dots,n_{k})}{n_{k}}I_{2}\Big)
 \in\mathcal{S}_{[n_{1},\dots,n_{k}]}$$
 we  define a congruence $\equiv$ on
$\mathcal{S}_{[n_{1},\dots,n_{k}]}$ with respect to multiplication and the adjoint operation as
 $$\mathbb{H}\equiv \mathbb{G}\ \Leftrightarrow\
 \mathbb{D}\mathbb{H}=\mathbb{D}\mathbb{G}\ (\mathrm{mod}\ \lcm(n_{1},\dots,n_{k})) \quad \textrm{for}
\quad \mathbb{H},\mathbb{G}\in\mathcal{S}_{[n_{1},\dots,n_{k}]}
 $$  and denote $[\mathbb{H}]$ the equivalence class of $\mathbb{H}$.

The \emph{symplectic group} $\Sp_{[n_{1},\dots,n_{k}]}$ is now defined as
 \begin{equation} \label{group}
 \Sp_{[n_{1},\dots,n_{k}]}=
 \set{[\mathbb{H}]\ }{ \mathbb{H}\in \mathcal{S}_{[n_{1},\dots,n_{k}]}, \ \ \mathbb{H}^{\ast}\mathbb{J}\mathbb{H}\equiv \mathbb{J}}.
 \end{equation}
where $
 \mathbb{J}=\diag(J_{2},\dots,J_{2})\in\mathcal{S}_{[n_{1},\dots,n_{k}]}$
and $J_{2}=\left(\begin{array}{cc}0&1\\-1&0\end{array}\right)$.
\end{definition}

\section{Clifford groups and exact sequences}

The  short exact sequences
\begin{equation*}
     1 \rightarrow \H_{(n_1,\dots,n_k)} \stackrel{\nu}{\longrightarrow} \CC_{(n_1,\dots,n_k)}
\stackrel{\pi}{\longrightarrow}\Sp_{[n_1,\dots,n_{k}]}  \rightarrow 1 .
\end{equation*}
and
\begin{equation*}
    1 \rightarrow \overline{\H}_{(n_1,\dots,n_k)} \stackrel{\overline{\nu}}{\longrightarrow} \overline{\CC}_{(n_1,\dots,n_k)}
\stackrel{\overline{\pi}}{\longrightarrow}\Sp_{[n_1,\dots,n_{k}]}  \rightarrow 1 .
\end{equation*}
comming from Theorem \ref{theorem_1}, where $\nu$ and $\overline{\nu}$ are inclusions and $\pi(U)=\overline{\pi}(Ad_U)$ for $U\in \CC_{(n_1,\dots,n_k)}$,
will be  called \emph{natural}.

\

As a basic fact for our further considerations we will use the following well known property of exact sequences from the group theory:

\

\emph{The Clifford group (the projective Clifford group, resp.) has a structure of a semidirect product related to its natural exact sequence if and only this sequence is right splitting.}

\

The right splitting of the natural exact sequence for the Clifford group $\CC_{(n_1,\dots,n_k)}$ means that there is a group homomorphism $\vartheta: \Sp_{[n_1,\dots,n_{k}]}\to \CC_{(n_1,\dots,n_k)}$ such that $\pi\circ\vartheta=id_{\Sp_{[n_1,\dots,n_{k}]}}$ (and similarly for the projective Clifford group).

\

The next statement recalls the fact that two multipartite systems with isomorphic configuration space have also isomorphic natural exact sequences of their Clifford groups (projective Clifford groups, resp.). An interesting fact is that the isomorphism is provided via a permutation matrix. For completeness, we provide a short proof.

\begin{proposition}\label{equivalent}
Let $k,\ell\geq1$ and $n_1,\dots,n_k,m_1\dots,m_\ell\geq 2$ be positive integers such that  $\Z_{n_1}\oplus\dots\oplus\Z_{n_k}\cong\Z_{m_1}\oplus\cdots\oplus\Z_{m_\ell}$. Then:
\begin{enumerate}
 \item The natural exact sequences of $\CC_{(n_1,\dots,n_k)}$ and $\CC_{(m_1,\dots,m_\ell)}$ are isomorphic via a permutation matrix. In particular, $\CC_{(n_1,\dots,n_k)}$ is a natural semidirect product if and only if $\CC_{(m_1,\dots,m_\ell)}$ is a natural semidirect product.
 \item The natural exact sequences of $\overline{\CC}_{(n_1,\dots,n_k)}$ and $\overline{\CC}_{(m_1,\dots,m_\ell)}$ are isomorphic. In particular, $\overline{\CC}_{(n_1,\dots,n_k)}$ is a natural semidirect product if and only if $\overline{\CC}_{(m_1,\dots,m_\ell)}$ is a natural semidirect product.
\end{enumerate}
\end{proposition}
\begin{proof}
First, let $a$ and $b$ be positive integers with $\gcd(a,b)=1$. Then $\Z_{ab}\cong\Z_a\oplus\Z_b$ and by the Chinese remainer theorem (see e.g. \cite{BengZycz}) there is a permutation matrix $R$ such that $R\H_{(n,m)}R^{-1}=\H_{(nm)}$. In this way we can decompose $n_i$ into product of power of primes. Hence we obtain a permutation matrix $R_1$ such that $R_1\H_{(n_1,\dots,n_k)}R_1^{-1}=\H_{(q_1,\dots,q_p)}$, where $q_1,\dots,q_p$ are powers of primes and $\Z_{n_1}\oplus\dots\oplus\Z_{n_k}\cong\Z_{q_1}\oplus\cdots\oplus\Z_{q_p}$.  Since every finite abelian group has (up to an ordering) a unique decomposition  into a direct sum of cyclic groups of prime power order, we obtain similarly a permutation matrix $R_2$ such that $R_2\H_{(m_1,\dots,m_\ell)}R_2^{-1}=\H_{(q_{\sigma(1)},\dots,q_{\sigma(p)})}$ for some permutation $\sigma$ on $p$ elements. Now, there is a permutation matrix that permutes the order of matrices in the Kronecker product (see e.g. \cite{HornJohn}). Hence $R_3\H_{(q_1,\dots,q_p)}R_3^{-1}=\H_{(q_{\sigma(1)},\dots,q_{\sigma(p)})}$ for some permutation matrix $R_3$. Consequently, $U\H_{(n_1,\dots,n_k)}U^{-1}=\H_{(m_1,\dots,m_\ell)}$ for the permutation matrix $U=R_2^{-1}R_3R_1$.

We obtain in this way the following commutative diagram that provides the isomorphism of exact sequences (where the last vertical arrow is uniquely determined by the others):
 \begin{equation*}
\begin{tikzcd}\label{diagram_2}
1 \arrow{r}{}& \H_{(n_1,\dots,n_k)} \arrow{d}{U(\cdot)U^{-1}}  \arrow{r}{\nu} & \CC_{(n_1,\dots,n_k)} \arrow{d}{U(\cdot)U^{-1}} \arrow{r}{\pi}  & \Sp_{[n_1,\dots,n_k]} \arrow{d}{} \arrow{r}{} & 1\\
1 \arrow{r}{} & \H_{(m_1,\dots,m_\ell)} \arrow{r}{\nu'} &\CC_{(m_1,\dots,m_\ell)} \arrow{r}{\pi'}  & \Sp_{[m_1,\dots,m_\ell]}  \arrow{r}{}  & 1
\end{tikzcd}
\end{equation*}
The projective case is done similarly.
\end{proof}

We will also need this simple observation.

\begin{proposition}\label{linear_to_projective}
Let $k\geq1$ and $n_1,\dots,n_k\geq 2$ be positive integers. If the Clifford group $\CC_{(n_1,\dots,n_k)}$ is a natural semidirect product, then the projective Clifford group $\overline{\CC}_{(n_1,\dots,n_k)}$ is a natural semidirect product as well.
\end{proposition}
\begin{proof}
Consider the commutative diagram of natural short exact sequences
 \begin{equation*}
\begin{tikzcd}
1 \arrow{r}{}& \H_{(n_1,\dots,n_k)} \arrow{r}{} & \CC_{(n_1,\dots,n_k)} \arrow{d}{p} \arrow{r}{\pi}  & \Sp_{[n_1,\dots,n_k]} \arrow[d, equal] \arrow{r}{} & 1\\
1 \arrow{r}{} & \overline{\H}_{(n_1,\dots,n_k)} \arrow{r}{} &\overline{\CC}_{(n_1,\dots,n_k)} \arrow{r}{\overline{\pi}}  & \Sp_{[n_1,\dots,n_k]}  \arrow{r}{} & 1  &
\end{tikzcd}
\end{equation*}
where $p:\CC_{(n_1,\dots,n_k)}\to\overline{\CC}_{(n_1,\dots,n_k)}$, $U\mapsto Ad_U$ is a projection and $\pi=\overline{\pi}\circ p$.

If the Clifford group $\CC_{(n_1,\dots,n_k)}$ is a natural semidirect product then there is a group homomorphism $\vartheta:\Sp_{[n_1,\dots,n_k]}\to \CC_{(n_1,\dots,n_k)}$ such that $\pi\circ\vartheta=id_{\Sp_{[n_1,\dots,n_k]}}$.
 It follows that $\overline{\pi}\circ (p\circ \vartheta)=\pi\circ\vartheta=id_{\Sp_{[n_1,\dots,n_k]}}$ and the lower exact sequence is therefore right splitting with the homomorphism $p\circ\vartheta:\Sp_{[n_1,\dots,n_k]}\to \overline{\CC}_{(n_1,\dots,n_k)}$. Hence the projective Clifford group $\overline{\CC}_{(n_1,\dots,n_k)}$ is a natural semidirect product.
\end{proof}

\begin{lemma}\label{decomp}
Let $A,A'$ be complex regular matrices $n\times n$ and $B,B'$ be complex regular matrices $m\times m$. If $A\otimes B=A'\otimes B'$, then there is $\lambda\in\C\setminus\{0\}$ such that $A'=\lambda A$ and $B'=\lambda^{-1}B$.
\end{lemma}
\begin{proof}
 Regular matrices are invertible, so from $A\otimes B=A'\otimes B'$ it follows that $I_{nm}=(A'\otimes B')(A\otimes B)^{-1}=(A'\otimes B')(A^{-1}\otimes B^{-1})=(A'A^{-1}\otimes B'B^{-1})$. As $I_{nm}$ is diagonal, it follows that $A'A^{-1}$ and $B'B^{-1}$ have to be diagonal also and, in fact, there is $\lambda\in\C\setminus\{0\}$ such that $A'A^{-1}=\lambda I_n$ and $B'B^{-1}=\lambda^{-1}I_m$. Hence  $A'=\lambda A$ and $B'=\lambda^{-1}B$.
\end{proof}

\begin{lemma}\label{lemma_1.3}
Let $k\geq 2$, $1\leq \ell< k$ and $n_1,\dots,n_k\geq 2$  be positive integers. Let $$\overline{\pi}: \overline{\CC}_{(n_1,\dots,n_k)}\to \Sp_{[n_1,\cdots,n_{k}]}$$
and
$$\overline{\pi'}: \overline{\CC}_{(n_1,\dots,n_\ell)}\to \Sp_{[n_1,\cdots,n_{\ell}]}$$
be projections in the natural exact sequences for the corresponding projective Clifford groups.

Let for  $[\mathbb{H}]\in \Sp_{[n_1,\cdots,n_{\ell}]}$ be $U\in \U(n_1\cdots n_k)$ such that $Ad_{U}\in \overline{\CC}_{(n_1,\dots,n_k)}$ and
$\left[\begin{smat2}{\mathbb{H}}{0}{0}{I_{2(k-\ell)}}\end{smat2}\right]
=\overline{\pi}(Ad_U)\in\Sp_{[n_{1},\dots,n_{k}]}$. Then there are unique, up to a non-zero scalar multiples, matrices $U_1\in \U(n_1\cdots n_{\ell})$ and $U_2\in \U(n_{\ell+1}\cdots n_k)$ such that $Ad_{U_1}\in  \overline{\CC}_{(n_1,\dots,n_\ell)}$, $U=U_1\otimes U_2$ and $[\mathbb{H}]=\overline{\pi'}(Ad_{U_1})\in \Sp_{[n_1,\cdots,n_{\ell}]}$.
\end{lemma}
\begin{proof}
Assume that that there are such $[\mathbb{H}]\in \Sp_{[n_1,\cdots,n_{\ell}]}$ and $U\in \U(n_1\cdots n_k)$. Since $\overline{\pi'}$ is an epimorphism, there is  $V_1\in \U(n_1\cdots n_\ell)$ such that $Ad_{V_1}\in \overline{\CC}_{(n_1,\dots,n_\ell)}$  and  $\overline{\pi'}(Ad_{U_1})=[\mathbb{H}]$. By the acting of $Ad_{V_1}$ on $\overline{\H}_{(n_1,\dots,n_\ell)}$, as inner automorphism, it immediately follows, by Theorem \ref{theorem_1}, that  $$\overline{\pi}\left(Ad_{V_1\otimes I_{n_{\ell+1}\cdots n_k}}\right)=\left[\begin{smat2}{\mathbb{H}}{0}{0}{I_{2(k-\ell)}}\end{smat2}\right]
\in\Sp_{[n_{1},\dots,n_{k}]}.$$

Now, as $\overline{\pi}(Ad_U)=\overline{\pi}(Ad_{V_1\otimes I_{n_{\ell+1}\cdots n_k}})$, the elements $Ad_U$ and $Ad_{V_1\otimes I_{n_{\ell+1}\cdots n_k}}$ differ only up to an element from $\ker(\overline{\pi})=\overline{\H}_{(n_1,\dots,n_k)}$. In particular, there are $i_1,\dots,i_n,j_1,\dots,j_n\in\Z$ and $W_1=P_{n_1}^{i_1}Q_{n_1}^{j_1}\otimes\cdots\otimes P_{n_\ell}^{i_\ell}Q_{n_\ell}^{j_\ell}$ and $W_2=P_{n_{\ell+1}}^{i_{\ell+1}}Q_{n_{\ell+1}}^{j_{\ell+1}}\otimes\cdots\otimes P_{n_k}^{i_k}Q_{n_k}^{j_k}$ such that $Ad_{U}=Ad_{V_1\otimes I_{n_{\ell+1}\cdots n_k}}\cdot Ad_{W_1\otimes W_2}$. Hence, there is $\lambda\in\C\setminus\{0\}$ such that $$U=\lambda(V_1\otimes I_{n_{\ell+1}\cdots n_k})\cdot (W_1\otimes W_2)=\lambda(V_1W_1\otimes W_2).$$ As both $V_1W_1$ and $W_2$ are unitary matrices, so is $V_1W_1\otimes W_2$. And since $U$ is also a unitary matrix, it follows that $|\lambda|=1$ and the matrices $U_1=V_1W_1$ and $U_2=\lambda W_2$ are therefore unitary as well. We have thus $U=U_1\otimes U_2$ and, by Lemma \ref{decomp}, the matrices $U_1$ and $U_2$ are unique up a non-zero scalar multiple. Finally, $Ad_{U_1}=Ad_{V_1}Ad_{W_1}\in \overline{\CC}_{(n_1,\dots,n_{\ell})}$.
\end{proof}

The next lemma is a key argument how to reduce the problem of the non-existence of the natural semidirect structure of Clifford group.

\begin{proposition}\label{short}
Let $k\geq 2$, $1\leq \ell< k$ and $n_1,\dots,n_k\geq 2$  be positive integers.
If the projective Clifford group $\overline{\CC}_{(n_1,\dots,n_k)}$ is a natural semidirect product, then the projective Clifford group $\overline{\CC}_{(n_1,\dots,n_\ell)}$ is also a natural semidirect product.
\end{proposition}
\begin{proof}
Consider the natural exact sequences for the corresponding projective Clifford groups:
\begin{equation*} \label{exact4}
     1 \rightarrow \overline{\H}_{(n_1,\dots,n_k)} \stackrel{\overline{\nu}}{\longrightarrow} \overline{\CC}_{(n_1,\dots,n_k)}
\stackrel{\overline{\pi}}{\longrightarrow} \Sp_{[n_1,\dots,n_{k}]}  \rightarrow 1 .
\end{equation*}
\begin{equation*} \label{exact5}
     1 \rightarrow \overline{\H}_{(n_1,\dots,n_\ell)} \stackrel{\overline{\nu'}}{\longrightarrow} \overline{\CC}_{(n_1,\dots,n_\ell)}
\stackrel{\overline{\pi'}}{\longrightarrow} \Sp_{[n_1,\dots,n_{\ell}]}  \rightarrow 1 .
\end{equation*}

Assume that there is a right splitting homomorphisms $\vartheta:\Sp_{[n_1,\dots,n_{k}]} $ $\to  \overline{\CC}_{(n_1,\dots,n_k)}$ such that $\overline{\pi}\circ\vartheta=id_{\Sp_{[n_1,\dots,n_{k}]}}$. We construct a splitting homomorphism $\vartheta':\Sp_{[n_1,\dots,n_{\ell}]} \to  \overline{\CC}_{(n_1,\dots,n_\ell)}$ such that $\overline{\pi'}\circ\vartheta'=id_{\Sp_{[n_1,\dots,n_{\ell}]}}$.

According the Definition \ref{def_11}, the map $\psi:\Sp_{[n_1,\dots,n_{\ell}]}\to \Sp_{[n_1,\dots,n_{k}]}$, $\psi\left([\mathbb{H}]\right)=\left[\begin{smat2}{\mathbb{H}}{0}{0}{I_{2(k-\ell)}}\end{smat2}\right]$ is clearly a group monomorphisms. As $\overline{\pi}$ is an epimorphism, for given $[\mathbb{H}]\in \Sp_{[n_1,\dots,n_{\ell}]}$ there is $U\in \U(n_1\cdots n_k)$ such that $\overline{\pi}(Ad_U)=\psi\left([\mathbb{H}]\right)$.
By Lemma \ref{lemma_1.3}, there are, up to a non-zero scalar multiples, matrices $U_1\in \U(n_1\cdots n_{\ell})$ and $U_2\in \U(n_{\ell+1}\cdots n_k)$ such that $Ad_{U_1}\in  \overline{\CC}_{(n_1,\dots,n_{\ell})}$, $U=U_1\otimes U_2$ and $[\mathbb{H}]=\overline{\pi'}(Ad_{U_1})$.
Now we set $\vartheta'\left([\mathbb{H}]\right)=Ad_{U_1}$. As $U_1$ is unique, up to a scalar multiple, the map $\vartheta'$ is well defined. Moreover, $\overline{\pi'}(\vartheta'([\mathbb{H}]))=\overline{\pi'}(Ad_{U_1})=[\mathbb{H}]$ and, therefore, $\overline{\pi'}\circ\vartheta'=id_{\Sp_{[n_1,\cdots,n_{\ell}]}}$.

Finally, similarly as above, for an arbitrary $[\mathbb{G}]\in \Sp_{[n_1,\cdots,n_{\ell}]}$, there is $V\in \U(n_1\cdots n_k)$ such that $\overline{\pi}(Ad_{V})=\psi\left([\mathbb{G}]\right)$ and matrices $V_1\in \U(n_1\cdots n_{\ell})$ and $V_2\in \U(n_{\ell+1}\cdots n_k)$ such that $Ad_{V_1}\in  \overline{\CC}_{(n_1,\dots,n_{\ell})}$, $V=V_1\otimes V_2$ and $[\mathbb{G}]=\overline{\pi'}(Ad_{V_1})$. For the matrix $UV\in \U(n_1\cdots n_k)$ there holds that $\overline{\pi}(Ad_{UV})=\overline{\pi}(Ad_{U})\overline{\pi}(Ad_{V})=\psi([\mathbb{H}])\psi([\mathbb{G}])=\psi([\mathbb{H}]\cdot[\mathbb{G}])$,  $UV=(U_1V_1)\otimes(U_2V_2)$ and $[\mathbb{H}]\cdot[\mathbb{G}]=\overline{\pi'}(Ad_{U_1})\overline{\pi'}(Ad_{V_1})=\overline{\pi'}(Ad_{U_1V_1})$. By the construction of the map $\vartheta'$, we obtain that $\vartheta'\left([\mathbb{H}]\cdot[\mathbb{G}]\right)=Ad_{U_1V_1}$. Therefore $\vartheta'\left([\mathbb{H}]\cdot[\mathbb{G}]\right)=Ad_{U_1V_1}=Ad_{U_1}Ad_{V_1}=\vartheta'([\mathbb{H}])\cdot\vartheta'([\mathbb{G}])$ and the map $\vartheta'$ is thus a group homomorphism.

This proves that the exact sequences for the group $\overline{\CC}_{(n_1,\dots,n_\ell)}$  is right splitting, hence the projective Clifford group $\overline{\CC}_{(n_1,\dots,n_\ell)}$ is a natural semidirect product.
\end{proof}

In the next section we show the \emph{non-existence} of the natural semidirect product of pojective Clifford group for certain systems composed of \emph{two} subsystems.

\section{Systems composed of two subsystems}

We are going to show that for  $n = 2\ (\mathrm{mod}\ 4)$ and $m = 2\ (\mathrm{mod}\ 4)$ the natural semidirect product structure of  $\overline{\CC}_{(n,m)}$
is prohibited (see Proposition \ref{composite_system_semidirect}), i.e., that the exact sequence

\begin{equation} \label{exact6}
     1 \rightarrow \overline{\H}_{(n,m)} \stackrel{\overline{\nu}}{\longrightarrow} \overline{\CC}_{(n,m)}
\stackrel{\overline{\pi}}{\longrightarrow} \Sp_{[n,m]}  \rightarrow 1 .
\end{equation}
is \emph{not} right splitting.

Such two-component composite system was
considered thoroughly in \cite{KorbTolar10}.

\

Let us recall that, by Theorem \ref{theorem_1}, for $U\in \CC_{(n,m)}$ and $\overline{\pi}(Ad_U)=[\mathbb{H}]\in \Sp_{[n,m]}$ and  for all $i_1,j_1,i_2,j_2\in\Z$ there holds
\begin{equation}\label{acting}
 Ad_U\left(P_{n}^{i_1}Q_{n}^{j_1}\otimes P_{m}^{i_2}Q_{m}^{j_2}\right)\sim P_{n}^{i'_1}Q_{n}^{j'_1}\otimes P_{m}^{i'_2}Q_{m}^{j'_2}
\end{equation}
where $\begin{pmatrix}i'_1\\[1mm] j'_1\\[1mm]i'_2\\[1mm]j'_2 \end{pmatrix}=\mathbb{H}\begin{pmatrix}i_1\\ j_1\\i_2\\j_2 \end{pmatrix}$.

\

Now, we will assume that for the exact sequence (\ref{exact6}) there exists a right splitting homomorphism
\begin{equation}\label{assumption}
\T:  \Sp_{[n,m]} \to \overline{\CC}_{(n,m)}
\end{equation}
such that $\overline{\pi}\circ\T=id_{\Sp_{[n,m]}}$, and we are going to obtain in this way a contradiction (see the proof of Proposition \ref{composite_system_semidirect}).

\

 For a more concise notation we will rather write $\T_{[\mathbb{H}]}$  than $\T([\mathbb{H}])$ for $[\mathbb{H}]\in \Sp_{[n,m]}$.

\

Now for $[\mathbb{H}]\in \Sp_{[n,m]}$ we can express the relation \ref{acting} in a simpler way:

For every $i_1,j_1,i_2,j_2\in\Z$ there holds
\begin{equation}\label{acting_2}
 \T_{[\mathbb{H}]}\left(P_{n}^{i_1}Q_{n}^{j_1}\otimes P_{m}^{i_2}Q_{m}^{j_2}\right)\sim P_{n}^{i'_1}Q_{n}^{j'_1}\otimes P_{m}^{i'_2}Q_{m}^{j'_2}
\end{equation}
where $\begin{pmatrix}i'_1\\[1mm] j'_1\\[1mm]i'_2\\[1mm]j'_2 \end{pmatrix}=\mathbb{H}\begin{pmatrix}i_1\\ j_1\\i_2\\j_2 \end{pmatrix}$.

Moreover, the elements $\T_{[\mathbb{H}]}\in \overline{\CC}_{(n,m)}$  are automorphism of $\H_{(n,m)}$ that are identical on the  center $Z(\H_{(n,m)})$ of $\H_{(n,m)}$, i.e. on the set $$Z(\H_{(n,m)})=\set{\lambda\cdot I_{nm}}{\lambda\in\U(nm)}.$$

\

In \cite{KorbTolar10} it was shown that
the symmetry group $\mathrm{Sp}_{[n,m]}$ is generated by five elements
$$ \alpha=\begin{matt2}{T_2}{0}{0}{I_2}\end{matt2},\ \  \beta=\begin{matt2}{J_2}{0}{0}{I_2}\end{matt2},
 \ \  \alpha'=\begin{matt2}{I_2}{0}{0}{T_2}\end{matt2},\ \ \beta'=\begin{matt2}{I_2}{0}{0}{J_2}\end{matt2}$$
$$  \text{and} \ \ \gamma=\begin{matt2}{I_2}{aE_{12}}{bE_{12}}{I_2}\end{matt2}$$
where  $a=\frac{n}{\gcd(n,m)}$, $b=\frac{m}{\gcd(n,m)}$ and $$I_2=\begin{mat2}{1}{0}{0}{1}\end{mat2}, \ \  T_2=\begin{mat2}{1}{1}{0}{1}\end{mat2}, \ \
J_2=\begin{mat2}{0}{-1}{1}{0}\end{mat2}, \ \ E_{12}=\begin{mat2}{0}{1}{0}{0}\end{mat2}.$$

 \begin{lemma}\label{first_form}
 Considering the relation (\ref{acting_2}), the automorphisms assigned to $\alpha,\beta,\alpha',\beta'$ and $\gamma$ via $\T$  take the form
\begin{displaymath}
\begin{array}{ll}
\begin{array}{ll}
\T_{\alpha}:
  & P_n\otimes I_m\  \mapsto\  \lambda_1(P_n\otimes I_m)\\
  & Q_n\otimes I_m\  \mapsto\  \lambda_2(P_nQ_n\otimes I_m)\\
  & I_n\otimes P_m\  \mapsto\  \lambda_3(I_n\otimes P_m) \\
  & I_n\otimes Q_m\  \mapsto\  \lambda_4(I_n\otimes Q_m)
\end{array}
&
\begin{array}{ll}
\T_{\alpha'}:
  & P_n\otimes I_m\  \mapsto\  \lambda'_1(P_n\otimes I_m)\\
  & Q_n\otimes I_m\  \mapsto\  \lambda'_2(Q_n\otimes I_m)\\
  & I_n\otimes P_m\  \mapsto\  \lambda'_3(I_n\otimes P_m) \\
  & I_n\otimes Q_m\  \mapsto\  \lambda'_4(I_n\otimes P_mQ_m)
\end{array}
\end{array}
\end{displaymath}

\begin{displaymath}
\begin{array}{ll}
\begin{array}{ll}
\T_{\beta}:
  & P_n\otimes I_m\  \mapsto\  \mu_1(Q^{-1}_n\otimes I_m)\\
  & Q_n\otimes I_m\  \mapsto\  \mu_2(P_n\otimes I_m)\\
  & I_n\otimes P_m\  \mapsto\  \mu_3(I_n\otimes P_m) \\
  & I_n\otimes Q_m\  \mapsto\  \mu_4(I_n\otimes Q_m)
\end{array}
&
\begin{array}{ll}
\T_{\beta'}:
  & P_n\otimes I_m\  \mapsto\  \mu'_1(P_n\otimes I_m)\\
  & Q_n\otimes I_m\  \mapsto\  \mu'_2(Q_n\otimes I_m)\\
  & I_n\otimes P_m\  \mapsto\  \mu'_3(I_n\otimes Q^{-1}_m) \\
  & I_n\otimes Q_m\  \mapsto\  \mu'_4(I_n\otimes P_m)
\end{array}
\end{array}
\end{displaymath}

\begin{displaymath}
\begin{array}{ll}
 \T_{\gamma}:
  & P_n\otimes I_m\  \mapsto\  \nu_1(P_n\otimes I_m)\\
  & Q_n\otimes I_m\  \mapsto\  \nu_2(Q_n\otimes P^b_m)\\
  & I_n\otimes P_m\  \mapsto\  \nu_3(I_n\otimes P_m) \\
  & I_n\otimes Q_m\  \mapsto\  \nu_4(P^a_n\otimes Q_m)
\end{array}
\end{displaymath}
for some $\lambda_i,\lambda'_i, \mu_i,\mu'_i, \nu_i\in U(1)$, $i=1,2,3,4$.
\end{lemma}

\

Now, as the starting point, the following algebraic relations between the generating elements of $\Sp_{[n,m]}$
will be proved.

\begin{proposition}\label{basic relations}
Set $id=\begin{matt2}{I_2}{0}{0}{I_2}\end{matt2}\in \Sp_{[n,m]}$. The following holds:
 \begin{enumerate}
  \item[(i)] $\alpha\alpha'=\alpha'\alpha$, $\beta\beta'=\beta'\beta$, $\alpha\beta'=\beta'\alpha$, $\beta\alpha'=\alpha'\beta$.
  \item[(ii)] $\gamma^{\gcd(n,m)}=id$, $\alpha^n=id$, $(\alpha')^m=id$, $\beta^4=id=(\beta')^4$.
  \item[(iii)] $\alpha\beta^2=\beta^2\alpha$, $\alpha'\beta'^2=\beta'^2\alpha'$.
  \item[(iv)] $(\alpha\alpha')\gamma=\gamma(\alpha\alpha')$, $(\beta^2\gamma)^2=id=(\gamma\beta^2)^2$, $(\beta'^2\gamma)^2=id=(\gamma\beta'^2)^2$.
 \end{enumerate}
\end{proposition}
\begin{proof}
(i), (ii) Obvious.

(iii) Follows easily from the equalities $\beta^2=\begin{matt2}{-I_2}{0}{0}{I_2}\end{matt2}$ and $\beta'^2=\begin{matt2}{I_2}{0}{0}{-I_2}\end{matt2}$.

(iv) We have $\alpha\alpha'=\begin{matt2}{T_2}{0}{0}{T_2}\end{matt2}$.
Hence $$(\alpha\alpha')\gamma=\begin{matt2}{T_2}{0}{0}{T_2}\end{matt2}\begin{matt2}{I_2}{aE_{12}}{bE_{12}}{I_2}\end{matt2}=\begin{matt2}{T_2}{aT_2E_{12}}{bT_2E_{12}}{T_2}\end{matt2}=\begin{matt2}{T_2}{aE_{12}}{bE_{12}}{T_2}\end{matt2}$$
$$\gamma(\alpha\alpha')=\begin{matt2}{I_2}{aE_{12}}{bE_{12}}{I_2}\end{matt2}\begin{matt2}{T_2}{0}{0}{T_2}\end{matt2}=\begin{matt2}{T_2}{aE_{12}T_2}{bE_{12}T_2}{T_2}\end{matt2}=\begin{matt2}{T_2}{aE_{12}}{bE_{12}}{T_2}\end{matt2}$$ since $T_2E_{12}=E_{12}=E_{12}T_2$ and we obtain $(\alpha\alpha')\gamma=\gamma(\alpha\alpha')$.

Further, $\beta^2\gamma=\begin{matt2}{-I_2}{0}{0}{I_2}\end{matt2}\begin{matt2}{I_2}{aE_{12}}{bE_{12}}{I_2}\end{matt2}=\begin{matt2}{-I_2}{-aE_{12}}{bE_{12}}{I_2}\end{matt2}$ and $$(\beta^2\gamma)^2=\begin{matt2}{-I}{-aE_{12}}{bE_{12}}{I}\end{matt2}\begin{matt2}{-I_2}{-aE_{12}}{bE_{12}}{I_2}\end{matt2}=\begin{matt2}{I_2}{0}{0}{I_2}\end{matt2}$$ since $E_{12}E_{12}=0$. The rest is similar.
\end{proof}

\bigskip

\subsection{Preliminary calculations}

In this part, by using the relation from Proposition \ref{basic relations}, we obtain conditions that fulfill the scalars  $\lambda_i,\lambda'_i, \mu_i,\mu'_i, \nu_i\in U(1)$, $i=1,2,3,4$.

\begin{lemma}\label{basic}
The following holds:
\begin{enumerate}
 \item[(i)] $\lambda'_1=1$, $\lambda_3=1$.
  \item[(ii)] $\mu'_1=1$, $\lambda_4=1$.
 \item[(iii)] $\mu_3=1$, $\lambda'_2=1$.
 \item[(iv)] $\mu'_2= 1$, $\mu_4= 1$.

\end{enumerate}
\end{lemma}
\begin{proof}
We will use the relations from Proposition \ref{basic relations}.

(i) From $\alpha\alpha'=\alpha'\alpha$ it follows that $\T_{\alpha}\T_{\alpha'}=\T_{\alpha'}\T_{\alpha}$.
\begin{align*}\T_{\alpha}\T_{\alpha'}:\
  & P_n\otimes I_m\  \mapsto\  \lambda'_1(P_n\otimes I_m)\ \mapsto\ \lambda_1\lambda'_1(P_n\otimes I_m)\\
  & Q_n\otimes I_m\  \mapsto\  \lambda'_2(Q_n\otimes I_m)\ \mapsto\ \lambda_2\lambda'_2(P_nQ_n\otimes I_m)\\
  & I_n\otimes P_m\  \mapsto\  \lambda'_3(I_n\otimes P_m)\ \mapsto\ \lambda_3\lambda'_3(I_n\otimes P_m)\\
  & I_n\otimes Q_m\  \mapsto\  \lambda'_4(I_n\otimes P_mQ_m)\ \mapsto\ \lambda_4\lambda_3\lambda'_{4}(I_n\otimes P_mQ_m)
\end{align*}
\begin{align*}\T_{\alpha'}\T_{\alpha}:\
  & P_n\otimes I_m\  \mapsto\  \lambda_1(P_n\otimes I_m)\ \mapsto\ \lambda_1\lambda'_1(P_n\otimes I_m)\\
  & Q_n\otimes I_m\  \mapsto\  \lambda_2(P_nQ_n\otimes I_m)\ \mapsto\ \lambda_2\lambda'_1\lambda'_2(P_nQ_n\otimes I_m)\\
  & I_n\otimes P_m\  \mapsto\  \lambda_3(I_n\otimes P_m)\ \mapsto\ \lambda_3\lambda'_3(I_n\otimes P_m)\\
  & I_n\otimes Q_m\  \mapsto\  \lambda_4(I_n\otimes Q_m)\ \mapsto\ \lambda_4\lambda'_{4}(I_n\otimes P_mQ_m)
\end{align*}
Hence $\lambda'_1=1$ and $\lambda_3=1$.

\bigskip

(ii) From $\alpha\beta'=\beta'\alpha$ it follows that $\T_{\alpha}\T_{\beta'}=\T_{\beta'}\T_{\alpha}$.
\begin{align*}\T_{\alpha}\T_{\beta'}:\
  & P_n\otimes I_m\  \mapsto\  \mu'_1(P_n\otimes I_m)\ \mapsto\ \lambda_1\mu'_1(Q^{-1}_n\otimes I_m)\\
  & Q_n\otimes I_m\  \mapsto\  \mu'_2(Q_n\otimes I_m)\ \mapsto\ \lambda_2\mu'_2(P_nQ_n\otimes I_m)\\
  & I_n\otimes P_m\  \mapsto\  \mu'_3(I_n\otimes Q^{-1}_m)\ \mapsto\ \lambda_4^{-1}\mu'_3(I_n\otimes Q^{-1}_m)\\
  & I_n\otimes Q_m\  \mapsto\  \mu'_4(I_n\otimes P_m)\ \mapsto\ \lambda_3\mu'_4(I_n\otimes P_m)
\end{align*}
\begin{align*}\T_{\beta'}\T_{\alpha}:\
  & P_n\otimes I_m\  \mapsto\  \lambda_1(P_n\otimes I_m)\ \mapsto\ \lambda_1\mu'_1(P_n\otimes I_m)\\
  & Q_n\otimes I_m\  \mapsto\  \lambda_2(P_nQ_n\otimes I_m)\ \mapsto\ \lambda_2\mu'_1\mu'_2(P_nQ_n\otimes I_m)\\
  & I_n\otimes P_m\  \mapsto\  \lambda_3(I_n\otimes P_m)\ \mapsto\ \lambda_3\mu'_3(I_n\otimes Q^{-1}_m)\\
  & I_n\otimes Q_m\  \mapsto\  \lambda_4(I_n\otimes Q_m)\ \mapsto\ \lambda_4\mu'_{4}(I_n\otimes P_m)
\end{align*}
Hence $\mu'_1=1$ and $\lambda_3=\lambda_4$. Therefore, by (i),  $\lambda_4=\lambda_3=1$.

\bigskip

(iii) From $\beta\alpha'=\alpha'\beta$ it follows that $\T_{\beta}\T_{\alpha'}=\T_{\alpha'}\T_{\beta}$.
\begin{align*}\T_{\beta}\T_{\alpha'}:\
  & P_n\otimes I_m\  \mapsto\  \lambda'_1(P_n\otimes I_m)\ \mapsto\ \lambda'_1\mu_1(Q^{-1}_n\otimes I_m)\\
  & Q_n\otimes I_m\  \mapsto\  \lambda'_2(Q_n\otimes I_m)\ \mapsto\ \lambda'_2\mu_2(P_n\otimes I_m)\\
  & I_n\otimes P_m\  \mapsto\  \lambda'_3(I_n\otimes P_m)\ \mapsto\ \lambda'_3\mu_3(I_n\otimes P_m)\\
  & I_n\otimes Q_m\  \mapsto\  \lambda'_4(I_n\otimes P_mQ_m)\ \mapsto\ \lambda'_4\mu_3\mu_4(I_n\otimes P_mQ_m)
\end{align*}
\begin{align*}\T_{\alpha'}\T_{\beta}:\
  & P_n\otimes I_m\  \mapsto\  \mu_1(Q^{-1}_n\otimes I_m)\ \mapsto\ \lambda'^{-1}_2\mu_1(Q^{-1}_n\otimes I_m)\\
  & Q_n\otimes I_m\  \mapsto\  \mu_2(P_n\otimes I_m)\ \mapsto\ \lambda'_1\mu_2(P_n\otimes I_m)\\
  & I_n\otimes P_m\  \mapsto\  \mu_3(I_n\otimes P_m)\ \mapsto\ \lambda'_3\mu_3(I_n\otimes P_m)\\
  & I_n\otimes Q_m\  \mapsto\  \mu_4(I_n\otimes Q_m)\ \mapsto\ \lambda'_4\mu_{4}(I_n\otimes P_mQ_m)
\end{align*}
Hence $\lambda'_1=\lambda'_2$ and $\mu_3=1$. Therefore, by (i), $\lambda'_2=\lambda'_1=1$.

\bigskip

(iv) From $\beta\beta'=\beta'\beta$ it follows that $\T_{\beta}\T_{\beta'}=\T_{\beta'}\T_{\beta}$.
\begin{align*}\T_{\beta}\T_{\beta'}:\
  & P_n\otimes I_m\  \mapsto\  \mu'_1(P_n\otimes I_m)\ \mapsto\ \mu_1\mu'_1(Q^{-1}_n\otimes I_m)\\
  & Q_n\otimes I_m\  \mapsto\  \mu'_2(Q_n\otimes I_m)\ \mapsto\ \mu_2\mu'_2(P_n\otimes I_m)\\
  & I_n\otimes P_m\  \mapsto\  \mu'_3(I_n\otimes Q^{-1}_m)\ \mapsto\ \mu_4^{-1}\mu'_3(I_n\otimes Q^{-1}_m)\\
  & I_n\otimes Q_m\  \mapsto\  \mu'_4(I_n\otimes P_m)\ \mapsto\ \mu_3\mu'_4(I_n\otimes P_m)
\end{align*}
\begin{align*}\T_{\beta'}\T_{\beta}:\
  & P_n\otimes I_m\  \mapsto\  \mu_1(Q^{-1}_n\otimes I_m)\ \mapsto\ \mu_1\mu'^{-1}_2(Q^{-1}_n\otimes I_m)\\
  & Q_n\otimes I_m\  \mapsto\  \mu_2(P_n\otimes I_m)\ \mapsto\ \mu_2\mu'_1(P_n\otimes I_m)\\
  & I_n\otimes P_m\  \mapsto\  \mu_3(I_n\otimes P_m)\ \mapsto\ \mu_3\mu'_3(I_n\otimes Q^{-1}_m)\\
  & I_n\otimes Q_m\  \mapsto\  \mu_4(I_n\otimes Q_m)\ \mapsto\ \mu_4\mu'_{4}(I_n\otimes P_m)
\end{align*}
Hence  $\mu'_1=\mu'_2$ and $\mu_3=\mu_4$. Therefore, by (ii) and (iii), $\mu'_2=\mu'_1=1$ and $\mu_4=\mu_3=1$.
\end{proof}

\bigskip

Lemma \ref{basic} allows now simplify the form of the assumed automorphisms as follows:
\begin{displaymath}
\begin{array}{cc}
\begin{array}{rl}
\T_{\alpha}:
  & P_n\otimes I_m\  \mapsto\  \lambda_1(P_n\otimes I_m)\\
  & Q_n\otimes I_m\  \mapsto\  \lambda_2(P_nQ_n\otimes I_m)\\
  & I_n\otimes P_m\  \mapsto\  I_n\otimes P_m \\
  & I_n\otimes Q_m\  \mapsto\  I_n\otimes Q_m
\end{array}
&
\begin{array}{rl}
\T_{\alpha'}:
  & P_n\otimes I_m\  \mapsto\  P_n\otimes I_m\\
  & Q_n\otimes I_m\  \mapsto\  Q_n\otimes I_m\\
  & I_n\otimes P_m\  \mapsto\  \lambda'_3(I_n\otimes P_m) \\
  & I_n\otimes Q_m\  \mapsto\  \lambda'_4(I_n\otimes P_mQ_m)
\end{array}
\end{array}
\end{displaymath}

\begin{displaymath}
\begin{array}{cc}
\begin{array}{rl}
\T_{\beta}:
  & P_n\otimes I_m\  \mapsto\  \mu_1(Q^{-1}_n\otimes I_m)\\
  & Q_n\otimes I_m\  \mapsto\  \mu_2(P_n\otimes I_m)\\
  & I_n\otimes P_m\  \mapsto\  I_n\otimes P_m \\
  & I_n\otimes Q_m\  \mapsto\  I_n\otimes Q_m
\end{array}
&
\begin{array}{rl}
\T_{\beta'}:
  & P_n\otimes I_m\  \mapsto\  P_n\otimes I_m\\
  & Q_n\otimes I_m\  \mapsto\  Q_n\otimes I_m\\
  & I_n\otimes P_m\  \mapsto\  \mu'_3(I_n\otimes Q^{-1}_m) \\
  & I_n\otimes Q_m\  \mapsto\  \mu'_4(I_n\otimes P_m)
\end{array}
\end{array}
\end{displaymath}

\begin{displaymath}
\begin{array}{rl}
\begin{array}{rl}
 \T_{\gamma}:
  & P_n\otimes I_m\  \mapsto\  \nu_1(P_n\otimes I_m)\\
  & Q_n\otimes I_m\  \mapsto\  \nu_2(Q_n\otimes P^b_m)\\
  & I_n\otimes P_m\  \mapsto\  \nu_3(I_n\otimes P_m) \\
  & I_n\otimes Q_m\  \mapsto\  \nu_4(P^a_n\otimes Q_m)
\end{array}
&
\begin{array}{rl}
\T_{\alpha\alpha'}:
  & P_n\otimes I_m\  \mapsto\ \lambda_1(P_n\otimes I_m)\\
  & Q_n\otimes I_m\  \mapsto\ \lambda_2(P_nQ_n\otimes I_m)\\
  & I_n\otimes P_m\  \mapsto\ \lambda'_3(I_n\otimes P_m)\\
  & I_n\otimes Q_m\   \mapsto\ \lambda'_{4}(I_n\otimes P_mQ_m)
\end{array}
\end{array}
\end{displaymath}

\

\begin{lemma}\label{automorphisms}
\begin{displaymath}
\begin{array}{ll}
\begin{array}{ll}
\T_{\beta^2}:
  & P_n\otimes I_m\  \mapsto\   \mu_1\mu^{-1}_2(P^{-1}_n\otimes I_m)\\
  & Q_n\otimes I_m\ \mapsto\ \mu_1\mu_2(Q^{-1}_n\otimes I_m)\\
  & I_n\otimes P_m\  \mapsto\  I_n\otimes P_m\\
  & I_n\otimes Q_m\   \mapsto\ I_n\otimes Q_m
\end{array}
&
\begin{array}{ll}
\T_{\beta'^2}:
  & P_n\otimes I_m\  \mapsto\ P_n\otimes I_m\\
  & Q_n\otimes I_m\ \mapsto\ Q_n\otimes I_m\\
  & I_n\otimes P_m\  \mapsto\ \mu'_3\mu'^{-1}_4(I_n\otimes P^{-1}_m)\\
  & I_n\otimes Q_m\  \mapsto\ \mu'_3\mu'_4(I_n\otimes Q^{-1}_m)
\end{array}
\end{array}
\end{displaymath}
\end{lemma}
\begin{proof}
\begin{align*}\T_{\beta^2}=\T_{\beta}\cdot \T_{\beta}:\
  & P_n\otimes I_m\  \mapsto\  \mu_1(Q^{-1}_n\otimes I_m)\ \mapsto\ \mu_1\mu^{-1}_2(P^{-1}_n\otimes I_m)\\
  & Q_n\otimes I_m\  \mapsto\  \mu_2(P_n\otimes I_m)\ \mapsto\ \mu_1\mu_2(Q^{-1}_n\otimes I_m)\\
  & I_n\otimes P_m\  \mapsto\  I_n\otimes P_m\ \mapsto\ I_n\otimes P_m\\
  & I_n\otimes Q_m\  \mapsto\  I_n\otimes Q_m\ \mapsto\ I_n\otimes Q_m
\end{align*}
\begin{align*}\T_{\beta'^2}=\T_{\beta'}\cdot \T_{\beta'}:\
  & P_n\otimes I_m\  \mapsto\  P_n\otimes I_m\ \mapsto\ P_n\otimes I_m\\
  & Q_n\otimes I_m\  \mapsto\  Q_n\otimes I_m\ \mapsto\ Q_n\otimes I_m\\
  & I_n\otimes P_m\  \mapsto\  \mu'_3(I_n\otimes Q^{-1}_m)\ \mapsto\ \mu'_3\mu'^{-1}_4(I_n\otimes P^{-1}_m)\\
  & I_n\otimes Q_m\  \mapsto\  \mu'_4(I_n\otimes P_m)\ \mapsto\ \mu'_3\mu'_4(I_n\otimes Q^{-1}_m)
\end{align*}
\end{proof}

\begin{lemma}\label{lemma 2.3}
The following holds:
\begin{enumerate}
 \item[(i)] $\lambda_1=\pm1$, $\frac{\mu_2}{\mu_1}=\lambda^2_2\omega_n$
 \item[(ii)] $\lambda'_3=\pm1$, 
 \item[(iii)] $\nu_1=(\lambda'_3)^b$, $\nu_3=\lambda_1^a$.
 \item[(iv)] $\nu^b_3=1$, $\nu^2_4=\nu_1^{a}(\frac{\mu_2}{\mu_1})^a$.
 \item[(v)] $\nu^a_1=1$.
\end{enumerate}
where $\omega_n=e^{\frac{2\pi\mathbf{i}}{n}}\in\C$.
\end{lemma}
\begin{proof}
We will use the relations from Proposition \ref{basic relations} and the form of automorphisms from Lemma \ref{automorphisms}.

(i) From $\alpha\beta^2=\beta^2\alpha$ it follows that $\T_{\alpha}\T_{\beta^2}=\T_{\beta^2}\T_{\alpha}$.
\begin{align*}
\T_{\alpha}\T_{\beta^2}:\
  & P_n\otimes I_m\  \mapsto\  \mu_1\mu^{-1}_2(P^{-1}_n\otimes I_m)\ \mapsto\ \lambda^{-1}_1\mu_1\mu^{-1}_2(P^{-1}_n\otimes I_m)\\
  & Q_n\otimes I_m\  \mapsto\  \mu_1\mu_2(Q^{-1}_n\otimes I_m)\ \mapsto\ \lambda^{-1}_2\mu_1\mu_2(Q^{-1}_nP^{-1}_n\otimes I_m)\\
  & I_n\otimes P_m\  \mapsto\  I_n\otimes P_m\ \mapsto\ I_n\otimes P_m\\
  & I_n\otimes Q_m\  \mapsto\  I_n\otimes Q_m\ \mapsto\ I_n\otimes Q_m
\end{align*}
\begin{align*}\T_{\beta^2}\T_{\alpha}:\
  & P_n\otimes I_m\  \mapsto\  \lambda_1(P_n\otimes I_m)\ \mapsto\ \lambda_1\mu_1\mu^{-1}_2(P^{-1}_n\otimes I_m)\\
  & Q_n\otimes I_m\  \mapsto\  \lambda_2(P_nQ_n\otimes I_m)\ \mapsto\ \lambda_2\mu^2_1(P^{-1}_nQ^{-1}_n\otimes I_m)=\\
& \hspace{50mm} =\lambda_2\mu^2_1\omega_n(Q^{-1}_nP^{-1}_n\otimes I_m)\\
  & I_n\otimes P_m\  \mapsto\  I_n\otimes P_m\ \mapsto\ I_n\otimes P_m\\
  & I_n\otimes Q_m\  \mapsto\  I_n\otimes Q_m\ \mapsto\ I_n\otimes Q_m
\end{align*}
Hence $\lambda_1=\pm1$ and $\mu_2=\mu_1\lambda^2_2\omega_n$.

\bigskip

(ii) From $\alpha'\beta'^2=\beta'^2\alpha'$ it follows that $\T_{\alpha'}\T_{\beta'^2}=\T_{\beta'^2}\T_{\alpha'}$.
\begin{align*}\T_{\alpha'}\T_{\beta'^2}:\
 & P_n\otimes I_m\  \mapsto\  P_n\otimes I_m\ \mapsto\ P_n\otimes I_m\\
  & Q_n\otimes I_m\  \mapsto\  Q_n\otimes I_m\ \mapsto\ Q_n\otimes I_m\\
  & I_n\otimes P_m\  \mapsto\  \mu'_3(\mu'_4)^{-1}(I_n\otimes P^{-1}_m)\ \mapsto\ (\lambda'_3)^{-1}\mu'_3(\mu'_4)^{-1}(I_n\otimes P^{-1}_m)\\
  & I_n\otimes Q_m\  \mapsto\  \mu'_3\mu'_4(I_n\otimes Q^{-1}_m)\ \mapsto\ (\lambda'_4)^{-1}\mu'_3\mu'_4(I_n\otimes Q^{-1}_mP^{-1}_m)
\end{align*}
\begin{align*}\T_{\beta'^2}\T_{\alpha'}:\
  & P_n\otimes I_m\  \mapsto\  P_n\otimes I_m\ \mapsto\ P_n\otimes I_m\\
  & Q_n\otimes I_m\  \mapsto\  Q_n\otimes I_m\ \mapsto\ Q_n\otimes I_m\\
    & I_n\otimes P_m\  \mapsto\  \lambda'_3(I_n\otimes P_m)\ \mapsto\ \lambda'_3\mu'_3(\mu'_4)^{-1}(I_n\otimes P^{-1}_m)\\
  & I_n\otimes Q_m\  \mapsto\  \lambda'_4(I_n\otimes P_mQ_m)\ \mapsto\ \lambda'_4(\mu'_3)^2(I_n\otimes P^{-1}_m Q^{-1}_m)=\\
& \hspace{50mm}  =\lambda'_4(\mu'_3)^2\omega_m(I_n\otimes Q^{-1}_mP^{-1}_m)
\end{align*}
Hence $\lambda'_3=\pm1$. 

\bigskip

(iii) From $(\alpha\alpha')\gamma=\gamma(\alpha\alpha')$ it follows that $\T_{\alpha\alpha'}\T_{\gamma}=\T_{\gamma}\T_{\alpha\alpha'}$.
\begin{align*}\T_{\alpha\alpha'}\T_{\gamma}:\
  & P_n\otimes I_m\  \mapsto\  \nu_1(P_n\otimes I_m)\ \mapsto\ \nu_1\lambda_1(P_n\otimes I_m)\\
  & Q_n\otimes I_m\  \mapsto\  \nu_2(Q_n\otimes P_m^b)\ \mapsto\ \nu_2\lambda_2\lambda'^b_3(P_nQ_n\otimes P_m^b)\\
  & I_n\otimes P_m\  \mapsto\  \nu_3(I_n\otimes P_m)\ \mapsto\ \nu_3\lambda'_3(I_n\otimes P_m)\\
  & I_n\otimes Q_m\  \mapsto\  \nu_4(P_n^a\otimes Q_m)\ \mapsto\ \nu_4\lambda_1^a\lambda'_{4}(P_n^a\otimes P_mQ_m)
\end{align*}
\begin{align*}\T_{\gamma}\T_{\alpha\alpha'}:\
  & P_n\otimes I_m\  \mapsto\  \lambda_1(P_n\otimes I_m)\ \mapsto\ \lambda_1\nu_1(P_n\otimes I_m)\\
  & Q_n\otimes I_m\  \mapsto\  \lambda_2(P_nQ_n\otimes I_m)\ \mapsto\ \lambda_2\nu_1\nu_2(P_nQ_n\otimes P_m^b)\\
  & I_n\otimes P_m\  \mapsto\  \lambda'_3(I_n\otimes P_m)\ \mapsto\ \lambda'_3\nu_3(I_n\otimes P_m)\\
  & I_n\otimes Q_m\  \mapsto\  \lambda'_4(I_n\otimes P_mQ_m)\ \mapsto\ \lambda'_{4}\nu_3\nu_4(P_n^a\otimes P_mQ_m)
\end{align*}
Hence $\nu_1=(\lambda'_3)^b$ and $\nu_3=\lambda_1^a$.

\bigskip

(iv) First, let us find the form of  $\T_{\beta^2\gamma}$.

\begin{align*}\T_{\beta^2\gamma}=\T_{\beta^2}\T_{\gamma}:\
  & P_n\otimes I_m\  \mapsto\  \nu_1(P_n\otimes I_m)\ \mapsto\ \nu_1\mu_1\mu^{-1}_2(P_n^{-1}\otimes I_m)\\
  & Q_n\otimes I_m\  \mapsto\  \nu_2(Q_n\otimes P_m^b)\ \mapsto\ \nu_2\mu_1\mu_2(Q^{-1}_n\otimes P_m^b)\\
  & I_n\otimes P_m\  \mapsto\  \nu_3(I_n\otimes P_m)\ \mapsto\ \nu_3(I_n\otimes P_m)\\
  & I_n\otimes Q_m\  \mapsto\  \nu_4(P_n^a\otimes Q_m)\ \mapsto\ \nu_4\mu^a_1\mu^{-a}_2(P_n^{-a}\otimes Q_m)
\end{align*}
Now, from $(\beta^2\gamma)^2=id$ it follows that $(\T_{\beta^2\gamma})^2$ is an identical automorphism. We have
\begin{align*}(\T_{\beta^2\gamma})^2:\
  & P_n\otimes I_m\  \mapsto\  \nu_1\mu_1\mu^{-1}_2(P_n^{-1}\otimes I_m)\ \mapsto\ P_n\otimes I_m\\
  & Q_n\otimes I_m\  \mapsto\  \nu_2\mu_1\mu_2(Q^{-1}_n\otimes P_m^b)\ \mapsto\ \nu^b_3(Q_n\otimes I_m)\\
  & I_n\otimes P_m\  \mapsto\  \nu_3(I_n\otimes P_m)\ \mapsto\ \nu^2_3(I_n\otimes P_m)\\
  & I_n\otimes Q_m\  \mapsto\  \nu_4\mu^a_1\mu^{-a}_2(P_n^{-a}\otimes Q_m)\ \mapsto\ \nu^2_4\nu^{-a}_1\mu^a_1\mu^{-a}_2(I_n\otimes Q_m)
\end{align*}
and therefore $\nu^b_3=1$ and $\nu^2_4\nu^{-a}_1\mu^a_1\mu^{-a}_2=1$.

\bigskip

(v) Similarly, as in (iii), we first find the form of $\T_{\beta'^2\gamma}$.
\begin{align*}\T_{\beta'^2\gamma}=\T_{\beta'^2}\T_{\gamma}:\
  & P_n\otimes I_m\  \mapsto\  \nu_1(P_n\otimes I_m)\ \mapsto\ \nu_1(P_n\otimes I_m)\\
  & Q_n\otimes I_m\  \mapsto\  \nu_2(Q_n\otimes P_m^b)\ \mapsto\ \nu_2\mu'^b_3\mu'^{-b}_4(Q_n\otimes P_m^{-b})\\
  & I_n\otimes P_m\  \mapsto\  \nu_3(I_n\otimes P_m)\ \mapsto\ \nu_3\mu'_3\mu'^{-1}_4(I_n\otimes P^{-1}_m)\\
  & I_n\otimes Q_m\  \mapsto\  \nu_4(P_n^a\otimes Q_m)\ \mapsto\ \nu_4\mu'_3\mu'_4(P_n^{a}\otimes Q^{-1}_m)
\end{align*}

From $(\beta'^2\gamma)^2=id$ it follows that $(\T_{\beta'^2\gamma})^2$ is an identical automorphism. We have
\begin{align*}(\T_{\beta'^2\gamma})^2:\
  & P_n\otimes I_m\  \mapsto\  \nu_1(P_n\otimes I_m)\ \mapsto\ \nu^2_1(P_n\otimes I_m)\\
  & Q_n\otimes I_m\  \mapsto\  \nu_2\mu'^b_3\mu'^{-b}_4(Q_n\otimes P_m^{-b})\ \mapsto\ \nu^2_2\nu^{-b}_3\mu'^b_3\mu'^{-b}_4(Q_n\otimes I_m)\\
  & I_n\otimes P_m\  \mapsto\  \nu_3\mu'_3\mu'^{-1}_4(I_n\otimes P^{-1}_m)\ \mapsto\ I_n\otimes P_m\\
  & I_n\otimes Q_m\  \mapsto\  \nu_4\mu'_3\mu'_4(P_n^{a}\otimes Q^{-1}_m)\ \mapsto\ \nu^a_1(I_n\otimes Q_m)
\end{align*}
and thus $\nu^a_1=1$.
\end{proof}

\begin{proposition}\label{proposition 3.1}
For $k\in\N$ there holds:
\begin{enumerate}
 \item[(i)] \begin{align*}(\T_{\gamma})^k:\
  & P_n\otimes I_m\  \mapsto\  \nu^k_1(P_n\otimes I_m)\\
  & Q_n\otimes I_m\  \mapsto\  \nu^k_2\nu_3^{b\cdot\sum_{\ell=1}^{k-1}\ell}(Q_n\otimes P_m^{kb})\\
  & I_n\otimes P_m\  \mapsto\  \nu^k_3(I_n\otimes P_m)\\
  & I_n\otimes Q_m\  \mapsto\  \nu^k_4\nu_1^{a\cdot\sum_{\ell=1}^{k-1}\ell}(P_n^{ka}\otimes Q_m)
\end{align*}
 \item[(ii)] \begin{align*}(\T_{\alpha})^k:\
  & P_n\otimes I_m\  \mapsto\  \lambda^k_1(P_n\otimes I_m)\\
  & Q_n\otimes I_m\  \mapsto\  \lambda^k_2\lambda_1^{\sum_{\ell=1}^{k-1}\ell}(P_n^{k}Q_n\otimes I_m)\\
  & I_n\otimes P_m\  \mapsto\  I_n\otimes P_m\\
  & I_n\otimes Q_m\  \mapsto\  I_n\otimes Q_m
\end{align*}
\end{enumerate}
\end{proposition}
\begin{proof}
(i) We will proceed by induction. The case $k=1$ is obvious. From the induction step we obtain that
$$(\T_{\gamma})^{k+1}(Q_n\otimes I_m)=\T_\gamma\big((\T_{\gamma})^{k}(Q_n\otimes I_m)\big)=\T_\gamma\left(\nu^k_2\nu_3^{b\cdot\sum_{\ell=1}^{k-1}\ell}(Q_n\otimes P_m^{kb})\right)=$$
$$=\nu^k_2\nu_3^{b\cdot\sum_{\ell=1}^{k-1}\ell}\T_\gamma(Q_n\otimes I_m)\Big(\T_\gamma(I_n\otimes P_m)\Big)^{kb}=\nu^{k+1}_2\nu_3^{b\cdot\sum_{\ell=1}^{k}\ell}(Q_n\otimes P^b_m)(I_n\otimes P_m^{kb})=$$
$$=\nu^{k+1}_2\nu_3^{b\cdot\sum_{\ell=1}^{k}\ell}\Big(Q_n\otimes P_m^{(k+1)b}\Big)\ .$$
The rest is similar.

(ii) Again, we will proceed by  induction. The case $k=1$ is obvious. From the induction step we obtain that
$$(\T_{\alpha})^{k+1}(Q_n\otimes I_m)=\T_\alpha\big((\T_{\alpha})^{k}(Q_n\otimes I_m)\big)=\T_\alpha\left(\lambda^k_2\lambda_1^{\sum_{\ell=1}^{k-1}\ell}(P^k_nQ_n\otimes I_m)\right)=$$
$$=\lambda^k_2\lambda_1^{\sum_{\ell=1}^{k-1}\ell}\Big(\T_\alpha(P_n\otimes I_m)\Big)^{k}\T_\alpha(Q_n\otimes I_m)=\lambda^{k+1}_2\lambda_1^{\sum_{\ell=1}^{k}\ell}(P^k_n\otimes I_m)(P_nQ_n\otimes I_m)=$$
$$=\lambda^{k+1}_2\lambda_1^{\sum_{\ell=1}^{k}\ell}(P^{k+1}_nQ_n\otimes I_m)\ .$$
The rest is similar.

\end{proof}

\begin{lemma}\label{lemma 2.5}
The following holds:
\begin{enumerate}
 \item[(i)]  $\nu^d_2\nu_3^{b\binom{d}{2}}=1$, $\nu^d_4\nu_1^{a\binom{d}{2}}=1$.
 \item[(ii)] $\lambda^n_2\lambda_1^{\binom{n}{2}}=1$.
\end{enumerate}
where $d=\gcd(n,m)$ and $\binom{d}{2}$ is the binomial coefficient.
\end{lemma}
\begin{proof}
Follows from Propositions \ref{proposition 3.1} and \ref{basic relations}(ii).
\end{proof}

\bigskip

Equations for the scalar coefficients now can be simplified in the following way.

\begin{lemma}\label{lemma 3.4}
The following holds:
 \begin{enumerate}
    \item[(i)] $\lambda_1=\pm1$, $(\lambda_1)^{ab}=1$, $\lambda'_3=\pm1$, $(\lambda'_3)^{ab}=1$.
   \item[(ii)] $\nu^d_4=(\lambda'_3)^{\frac{abd(d-1)}{2}}$,  $\lambda^n_2=\lambda_1^{\frac{n(n-1)}{2}}$.
 \item[(iii)] $\nu^2_4=\lambda^{2a}_{2}\omega_d$.
 \end{enumerate}
 where  $d=\gcd(n,m)$ and $\omega_d=e^{\frac{2\pi \mathbf{i}}{d}}\in\C$.
\end{lemma}
\begin{proof}
(i) By Lemma \ref{lemma 2.3}(iii) and (iv), we have $1=\nu^b_3=(\lambda^a_1)^b$. Similarly, by Lemma \ref{lemma 2.3}(iii) and (v), we have $1=\nu^a_1=(\lambda'_3)^{ab}$. The rest is Lemma \ref{lemma 2.3}(i) and (ii).

(ii)  By Lemmas \ref{lemma 2.5}(i) and \ref{lemma 2.3}(iii), we have $1=\nu^d_4\nu_1^{a\binom{d}{2}}=\nu^d_4(\lambda'_3)^{ba\binom{d}{2}}$. By Lemma \ref{lemma 2.3}(ii), we have  $\lambda'_3=(\lambda'_3)^{-1}$. Hence $\nu^d_4=(\lambda'_3)^{ab\binom{d}{2}}$.

Similarly, by Lemmas \ref{lemma 2.5}(ii) and \ref{lemma 2.3} (i), we have $\lambda_1^{-1}=\lambda_1$ and $\lambda^n_2=(\lambda_1^{-1})^{\binom{n}{2}}=\lambda_1^{\frac{n(n-1)}{2}}$.

(iii) By Lemma \ref{lemma 2.3}(iv), (v) and (i), we have $\nu^2_4=\nu^{a}_1\left(\frac{\mu_2}{\mu_1}\right)^a=\left(\frac{\mu_2}{\mu_1}\right)^a=\left(\lambda^2_2\omega_n\right)^a=\lambda^{2a}_{2}\omega_d$, since $(\omega_n)^a=e^{\frac{2\pi\mathbf{i}}{n}a}=e^{\frac{2\pi\mathbf{i}}{d}}=\omega_d$. 
\end{proof}

\begin{lemma}\label{lemma 4.1}
 If $n=2\ (\mathrm{mod}\ 4)$ and $m=2\ (\mathrm{mod}\ 4)$, then  $\lambda_1=1$, $\lambda'_3=1$,  $\nu^d_4=1$ and $\lambda_2^n=1$, where $d=\gcd(n,m)$.
\end{lemma}
\begin{proof}
By Lemma \ref{lemma 3.4}(i), $\lambda_1=\pm1$ and $(\lambda_1)^{ab}=1$. Since $a=\frac{n}{\gcd(n,m)}$ and $b=\frac{m}{\gcd(n,m)}$ are odd, it follows that $\lambda_1=1$. Similarly, $\lambda'_3=1$. The rest follows easily from Lemma \ref{lemma 3.4}.
\end{proof}

\begin{proposition}\label{composite_system_semidirect}
 If $n=2\ (\mathrm{mod}\ 4)$ and $m=2\ (\mathrm{mod}\ 4)$, then the projective Clifford group $\overline{\CC}_{(n,m)}$ does not have the structure of a natural semidirect product.
\end{proposition}
\begin{proof}
Consider the natural exact sequence
 \begin{equation*} \label{exact7}
     1 \rightarrow \mathcal{P}_{(n,m)} \stackrel{}{\longrightarrow} \overline{\CC}_{(n,m)}
\stackrel{}{\longrightarrow} \Sp_{[n,m]}  \rightarrow 1 .
\end{equation*}

Assume, on the contrary, that there is a right splitting homomorphism
$$\T: \mathrm{Sp}_{[n,m]}\to \overline{\CC}_{(n,m)}$$
as in (\ref{assumption}) and let $\lambda_i,\lambda'_i, \mu_i,\mu'_i, \nu_i\in U(1)$, $i=1,2,3,4$ be as in Lemma \ref{first_form}. Then, by Lemmas \ref{lemma 4.1} and \ref{lemma 3.4}(iii), it holds that $$\nu^d_4=1,\ (\lambda_2)^n=1\ \text{and} \ \nu^2_4=(\lambda_2)^{2a}\omega_d$$ where $d=\gcd(n,m)$ and $\omega_d=e^{\frac{2\pi\mathbf{i}}{d}}$.  Since $d$ is even, there is $k\in\N$, $1\leq k<d$ such that $d=2k$. Hence, as $n=ad=2ak$, it follows that
$$1=\nu_4^d=(\nu^2_4)^k=(\lambda_2)^{2ka}\omega_d^k=(\lambda_2)^{n}\omega_d^k=\omega_d^{k}\ .$$
But this is a contradiction, because the order of $\omega_d$ is $d$,  and $1\leq k<d$.

Hence there is no right splitting homomorphism and the projective Clifford group $\overline{\CC}_{(n,m)}$ does not have the structure of a natural semidirect product.
\end{proof}

\section{The case of dimension divisible by four}\label{5}

\begin{theorem}\cite[Corollary 12.2]{KorbTolar23}\label{cyclic}
Let $n = 0\ (\mathrm{mod}\ 4)$. Then the projective Clifford group $\overline{\CC}_n$ is not a natural semidirect product.
\end{theorem}

\begin{theorem}\label{main_theorem_1}
Let $k\geq 2$ and $n_1,\dots,n_k\geq 2$  be positive integers. If $N=n_1\cdots n_k$ is divisible by four, then both Clifford group $\CC_{(n_1,\dots,n_k)}$ and the projective Clifford group $\overline{\CC}_{(n_1,\dots,n_k)}$ do not have the structure of a natural semidirect product.
\end{theorem}
\begin{proof}
If $N=n_1\cdots n_k$ is divisible by four, then either $n_{i_0}=0\ (\mathrm{mod}\ 4)$ for some $i_0$ or  $n_{j_0}=2\ (\mathrm{mod}\ 4)$ and $n_{j_1}=2\ (\mathrm{mod}\ 4)$ for some $j_0\neq j_1$. In view of Proposition \ref{equivalent}, we may therefore assume without loss of generality that:

either $n_{1}=0\ (\mathrm{mod}\ 4)$ and then our assertion holds by  Theorem \ref{cyclic} and Proposition \ref{short},

or  $n_1=2\ (\mathrm{mod}\ 4)$ and $n_2=2\ (\mathrm{mod}\ 4)$ and then our assertion follows from  Propositions \ref{composite_system_semidirect} and \ref{short}.
\end{proof}

\section{The case of even dimension not divisible by four}

In this section we show that if the dimension of a multipartite system with configation space $A=\Z_{n_1}\oplus\cdots\oplus\Z_{n_k}$ is divisible by $2$ but not by $4$, then the correspoding Clifford group (as well as the projective one) are indeed natural semidirect products. This fact seems to be not (broadly) known. The approach is simply to consider the decomposition $A\cong \Z_2 \oplus B$, where $B$ is an abelian group of odd order $m$. For $B$ it is known that the corresponding Clifford groups are semidirect product, i.e. their short exact sequences are right splitting (see e.g.\cite{DuttaPrasad}). For $\Z_2$ we write an explicit form of the right splitting homomorphism.

\begin{proposition}\label{two}
The Clifford group $\CC_2$ is a natural semidirect product.
\end{proposition}
\begin{proof}
Consider the natural short exact sequence \begin{equation*}
     1 \rightarrow \mathcal{H}_{2} \stackrel{\nu}{\longrightarrow} \mathcal{C}_2
\stackrel{\pi}{\longrightarrow} SL(2,\Z_2)  \rightarrow 1
\end{equation*}
We show that it is is right splitting.

The group $SL(2,\Z_2)$ is isomorphic to the permutation group  $S_3$ and has a presentation with  two generators $s=\begin{mat2}{0}{1}{1}{0}\end{mat2}$ and $t=\begin{mat2}{1}{1}{0}{1}\end{mat2}$ fulfilling the defining relations $s^2=t^2=(st)^3=id$.

Set two complex matrices $S=\tfrac{\sqrt{2}}{2}\begin{mat2}{1}{1}{1}{-1}\end{mat2}$ and $T=\tfrac{\sqrt{2}}{2}\begin{mat2}{-1}{-\mathbf{i}}{\mathbf{i}}{1}\end{mat2}$. It is easy to verify that $S,T\in \U(2)$. Further, it holds that

$$SP_2S^{-1}=Q_2,\ \ SQ_2S^{-1}=P_2$$
$$TP_2T^{-1}=-P_2,\ \ TQ_2T^{-1}=-\mathbf{i}\cdot P_2Q_2\ .$$

This implies that $S,T\in \mathcal{C}_2$ and that the corresponding matrices of automorphisms inducing by $S$ and $T$ are $s$ and $t$, i.e. $\pi(S)=s$ and $\pi(T)=t$.

Now, as we can easily check, $S^2=T^2=(ST)^3=I_2$, and therefore there is a (unique) homomorphism  $$\vartheta:SL(2,\Z_2)\to \CC_2$$ such that $\vartheta(s)=S$ and $\vartheta(t)=T$. Since $\pi(\vartheta(s))=s$ and $\pi(\vartheta(t))=t$ and $s,t$ are generators of $SL(2,\Z_2)$, we obtain that $\pi\circ\vartheta=id_{SL(2,\Z_2)}$.
Hence $\vartheta$ is the right splitting homomorphisms and, consequently, $\CC_2$ is a natural semidirect product.
\end{proof}

\begin{theorem}\label{main_theorem_2}
Let $k\geq 1$ and $n_1,\dots,n_k\geq 2$  be positive integers and $N=n_1\cdots n_k$. If $N=2$ (mod $4$), then both Clifford group $\CC_{(n_1,\dots,n_k)}$ and the projective Clifford group $\overline{\CC}_{(n_1,\dots,n_k)}$ are natural semidirect products.
\end{theorem}
\begin{proof}
If $N=2$, then the assertion follows from Propositions  \ref{equivalent} and \ref{two}.

If $N>2$, then  $N=2m$ for some odd integer $m\geq 3$. Then there are positive integers $\ell\geq 1$ and $m_1,\dots,m_\ell\geq 2$ such that $\Z_{n_1}\oplus\cdots\oplus\Z_{n_k}\cong\Z_2\oplus  B$ where  $B=\Z_{m_1}\oplus\cdots\oplus\Z_{m_\ell}$ an the order of $B$ is odd. According to Propositions  \ref{equivalent} and \ref{linear_to_projective}, it is enough to show that the exact sequence
 \begin{equation} \label{11}
	1 \rightarrow \H_{(2,m_1,\dots,m_\ell)} \stackrel{\nu}{\rightarrow} \CC_{(2,m_1,\dots,m_\ell)}
	\stackrel{\pi}{\rightarrow} \Sp_{[2,m_1,\dots,m_\ell]}  \rightarrow 1 .
\end{equation}
is right splitting. Since $\gcd(2,n_i)=1$ for every $i=1,\dots,k$, it follows, according to the Definition \ref{def_11} of $\Sp_{[2,m_1,\dots,m_\ell]}$, that every element of this group is of the form of block matrices $\left[\begin{array}{cc}\mathbb{H}&0\\ 0&\mathbb{G}\end{array}\right]$ where
$[\mathbb{H}]\in \Sp_2=SL(2,\Z_2)$ and $[\mathbb{G}]\in \Sp_{[m_1,\dots,m_k]}$. The abelian groups of odd order provide right splitting of the natural exact sequences for their Clifford groups (for the proof see e.g. \cite{DuttaPrasad}). Hence, as $B$ is of odd order, for the exact sequence
 \begin{equation*} \label{12}
	1 \rightarrow \H_{(m_1,\dots,m_\ell)} \stackrel{\nu'}{\rightarrow} \CC_{(m_1,\dots,m_\ell)}
	\stackrel{\pi'}{\rightarrow} \Sp_{[m_1,\dots,m_\ell]}  \rightarrow 1 .
\end{equation*}
there is a right splitting homomorphism $\vartheta_1: \Sp_{[m_1,\dots,m_\ell]}\to \CC_{(m_1,\dots,m_\ell)}$ such that $\pi'\circ\vartheta_1=id_{\Sp_{[m_1,\dots,m_\ell]}}$. Further, by Proposition \ref{two}, there is a right splitting homomorphism $\vartheta_2: SL(2,\Z_2)\to \CC_2$ such that $\pi''\circ\vartheta_2=id_{SL(2,\Z_2)}$ for the exact sequence
 \begin{equation*}
	1 \rightarrow \mathcal{H}_{2} \stackrel{\nu''}{\longrightarrow} \mathcal{C}_2
\stackrel{\pi''}{\longrightarrow} SL(2,\Z_2)  \rightarrow 1.
\end{equation*}

Now, we set  $$\vartheta\left(\left[\begin{array}{cc}\mathbb{H}&0\\ 0&\mathbb{G}\end{array}\right]\right)=\vartheta_2([\mathbb{H}])\otimes\vartheta_1([\mathbb{G}])$$ and it can be easily checked that $\vartheta$ is a right splitting homomorphism for the exact sequence (\ref{11}).
\end{proof}

\bigskip

\section{Conclusions}

In this paper we have fully answered a natural question coming already from classical mechanics:
is the Clifford group a semidirect product?

The answer (see Theorems \ref{main_theorem_1} and \ref{main_theorem_2}) is positive  if and only if  the dimension of the underlying Hilbert space  is \emph{not} divisible by four. Although in even dimensions the metaplectic group is usually considered  (as a covering of the symplectic group) and half-integer lattice states are used, the case when these are \emph{not} needed (i.e. the case of dimension $2m$ with $m$ beeing odd) seemed to be not known so far in full generality. In applications, usually higher power of two appear, so then the metaplectic group is necessary indeed, but the special  single case of the dimension $2m$, $m$ odd, sugests a question what of theoretical relations and applications it may allow.

\end{document}